\documentclass[aps,prd,onecolumn,nofootinbib,superscriptaddress,showkeys]{revtex4-2}

\usepackage[utf8]{inputenc}
\usepackage[T1]{fontenc}
\usepackage{lmodern}
\usepackage{amsmath,amssymb,amsfonts,amsthm}
\usepackage{graphicx}
\usepackage{hyperref}
\usepackage{xcolor}
\usepackage{pgfplots}
\usepackage[framemethod=tikz]{mdframed}
\usepackage{tikz}
\usetikzlibrary{arrows.meta, positioning, shapes, backgrounds}
\hypersetup{colorlinks=true,linkcolor=blue,citecolor=blue,urlcolor=magenta}

\newcommand{\Tr}{\mathrm{Tr}}
\newcommand{\I}{\hat I}
\theoremstyle{plain}
\newtheorem{theorem}{Theorem}[section]
\newtheorem{proposition}[theorem]{Proposition}
\newtheorem{lemma}[theorem]{Lemma}

\newtheorem{corollary}[theorem]{Corollary}

\begin{document}

\title{Dual Effective Field Theory formulation of Metric--Affine and Symmetric Teleparallel Gravity}
\author{Ginés R. Pérez Teruel}
\affiliation{Consellería de Educación, Cultura, Universidades y Empleo, Ministerio de Educación y Formación Profesional, Spain}
\email{gines.landau@gmail.com}

\begin{abstract}
We develop a unified algebraic and effective field theory (EFT) formulation for
non--Riemannian extensions of General Relativity with an independent connection.
For metric--affine $f(R,Q)$ gravity we show that the connection equations admit
an exact matrix solution, whose square--root structure generates a convergent
binomial/Neumann expansion in powers of the stress tensor $T_{\mu\nu}$.  For the
Eddington--inspired Born--Infeld (EiBI) theory we show that the connection can
be solved algebraically as well, and that its determinantal field equations
produce a parallel Neumann expansion with coefficients fixed by the underlying
determinant operator.  This allows us to rewrite the Einstein--like equations in
the auxiliary metric as an effective Einstein equation for $g_{\mu\nu}$ with a
local algebraic correction $(\Delta T)_{\mu\nu}$ that follows from a dual EFT
built from the invariants $\{T,\,T^2,\,T_{\mu\nu}T^{\mu\nu},\ldots\}$, organised
by a characteristic density scale.  We prove a convergence criterion based on
the spectral radius of $\hat T^\mu{}_\nu$ and interpret EiBI gravity as a
determinantal resummation of the same $T$--tower. Extending the framework to
symmetric teleparallel $f(Q)$ gravity, we identify the EFT coefficients in terms
of $f_Q$ and $f_{QQ}$ and present a background matching for
$f(Q)=Q+\alpha Q^2$.  The resulting dual EFT provides a common algebraic
language for metric--affine, Born--Infeld and non--metricity gravities.
\end{abstract}

\maketitle
\tableofcontents

\section{Introduction}
% ------------------------------------------------------------
Metric--affine (Palatini) gravity provides a geometrical framework in which
the metric $g_{\mu\nu}$ and the affine connection $\Gamma^\alpha{}_{\mu\nu}$ are treated as
independent dynamical variables~\cite{Ferraris1982,Hehl1995,SotiriouLiberati2006,Olmo2011,CapozzielloDeLaurentis2015}.
In its simplest realization, where the gravitational Lagrangian is a function
of the curvature invariants $R=g^{\mu\nu}R_{\mu\nu}(\Gamma)$ and
$Q=R_{\mu\nu}(\Gamma)R^{\mu\nu}(\Gamma)$,
the field equations remain second order and
admit an auxiliary--metric formulation
without introducing new propagating degrees of freedom~\cite{Ferraris1982,Olmo2005,Olmo2009,Olmo2011}.
This property contrasts with the metric $f(R)$ theories, whose fourth--order
dynamics can be recast as a scalar--tensor representation~\cite{SotiriouFaraoni2010,DeFeliceTsujikawa2010}.
In the Palatini case, by contrast, the independent connection can be
\emph{algebraically} eliminated, yielding an Einstein--like system with modified,
local matter couplings~\cite{Borowiec2012,OlmoRubieraGarcia2015,Borowiec2016,BeltranJimenez2018}.

Such theories have attracted interest both as minimal
extensions of General Relativity (GR) and as effective models arising from
Born--Infeld and string--inspired actions~\cite{DeserGibbons1998, Vollick2004, Vollick2005, BanadosFerreira2010, EscamillaRivera2012, Pani2012, Cho2021}. In particular, the Eddington--inspired Born--Infeld (EiBI) model 
\cite{BanadosFerreira2010,Scargill2012,Pani2012,Avelino2012,OlmoRubieraGarcia2013,OlmoRubieraGarcia2015}
belongs to the class of metric--affine theories whose connection field
equations are algebraic and yield an Einstein--like dynamics for an
auxiliary metric. It yields second--order field equations that
recover GR in vacuum but exhibit non--trivial matter--coupled dynamics,
with implications ranging from regular cosmological bounces
to compact--object phenomenology~\cite{Delsate2012,Pani2012,Sham2012,OlmoRubieraGarcia2013,Afonso2018,Pani2019,Jimenez2019}.

More recently, symmetric teleparallel theories based on the non--metricity scalar
$Q$~\cite{JimenezHeisenbergKoivisto2018,Jimenez2019,Koivisto2019,Capozziello2022}
have further emphasized the role of independent geometric variables and
motivated a systematic comparison between
curvature--, torsion--, and non--metricity--based formulations of gravity.
In the symmetric teleparallel limit, curvature and torsion vanish while
non--metricity remains nonzero, and General Relativity is recovered for
$f(Q)=Q$ (STEGR), whereas non--linear $f(Q)$ models provide flexible
dark--energy and early--universe scenarios.

The algebraic solvability of the Palatini connection equations,
originally observed by Ferraris and collaborators~\cite{Ferraris1982,Ferraris1988},
suggests that modifications of GR in this framework can be understood
as \emph{local, algebraic self--interactions of the matter stress tensor}
rather than as higher--derivative curvature corrections.
This observation has been explored in various contexts
(e.g.\ Refs.~\cite{Borowiec2012,OlmoRubieraGarcia2015,Borowiec2016,BeltranJimenez2018,BeltranJimenezKoivisto2020,Delhom2020,Afonso2018,OlmoRubieraGarcia2020,Baldazzi2021,HeisenbergEFT}),
yet its full analytic potential remains largely unexploited, and a unified
view including non--metricity--based gravity is still missing.

\smallskip

In this work we make this statement precise by constructing a
\emph{dual effective field theory (EFT)} representation for a broad class
of non--Riemannian gravities with an independent connection.
First, we revisit Ricci--based metric--affine $f(R,Q)$ models and the
Eddington--inspired Born--Infeld (EiBI) theory. 
Although EiBI is not a special case of $f(R,Q)$, it shares the same
Palatini algebraic structure and can be treated within the same analytic
framework.
Starting from the exact matrix solution associated with the Ricci tensor,
as derived in~\cite{PerezTeruel2013,PerezTeruel2014}, we show that, under a mild spectral condition on $\hat T^\mu{}_\nu$,
the square--root matrix entering the Palatini map admits a convergent
binomial series in powers of the stress--energy tensor $T_{\mu\nu}$.

We then prove a rigorous convergence lemma based on the spectral radius
of $\hat T^\mu{}_\nu$ and derive an auxiliary metric $h_{\mu\nu}$ as a
local power series in $\{g_{\mu\nu},T_{\mu\nu},(T^2)_{\mu\nu},\ldots\}$.
This construction allows us to map the Einstein--like equations for $h_{\mu\nu}$
to an effective Einstein equation for $g_{\mu\nu}$ with a local algebraic
correction $(\Delta T)_{\mu\nu}$, which in turn follows from an EFT action
built from the invariants $\{T,\,T^2,\,T_{\mu\nu}T^{\mu\nu},\ldots\}$ 
organized by a characteristic density scale. \\

Second, we extend this dual EFT framework beyond Ricci--based metric--affine
gravity.  On the one hand, we show that EiBI gravity corresponds to a
determinantal resummation of the same $T$--tower, providing a clear
dictionary between the Born--Infeld parameters and the EFT coefficients.
On the other hand, we outline a dual formulation for symmetric teleparallel
$f(Q)$ theories in which the non--metricity sector can be encoded in a
constitutive map $\hat T=\chi^\ast T$, and we derive the background matching
between the EFT coefficient $a_2(w)$ and the quadratic coupling of
$f(Q)=Q+\alpha Q^2$ in FLRW cosmology.  Throughout, we illustrate the formalism
with applications to cosmological fluids, electromagnetic fields, and compact stars.\\

Overall, the results reveal a common algebraic mechanism underlying
Palatini $f(R,Q)$ gravity, EiBI models, and symmetric teleparallel $f(Q)$
gravity: in the former two cases the independent connection is eliminated
by purely algebraic field equations, while in the symmetric teleparallel
case it can be fixed by gauge choice, and in all three frameworks the
resulting dynamics can be recast (at least perturbatively) as a local
dual EFT in powers of $T_{\mu\nu}$.

\section{Matrix formulation of Palatini $f(R,Q)$}

We consider the action
\begin{equation}
S[g,\Gamma,\psi]=\frac{1}{2\kappa^2}\!\int d^4x\,\sqrt{-g}\,f(R,Q)+S_m[g,\psi],
\end{equation}

where $R=g^{\mu\nu}R_{\mu\nu}(\Gamma)$, $Q=R_{\mu\nu}(\Gamma)R^{\mu\nu}(\Gamma)$, and with independent $g$ and $\Gamma$. Variation yields the matrix equation
\begin{equation}
2 f_Q \hat P^2 + f_R \hat P - \frac{1}{2} f \I = \kappa^2 \hat T,
\label{eq:matrix}
\end{equation}
where $(\hat P)^\mu{}_\nu \equiv R^{\mu\alpha}g_{\alpha\nu}$ and $(\hat T)^\mu{}_\nu \equiv T^\mu{}_\nu$.
Equation \eqref{eq:matrix} admits the exact solution
\begin{equation}
\hat P = -\frac{1}{4 f_Q}\Big(f_R \I - 2\sqrt{\alpha \I + \beta \hat T}\Big),\qquad
\alpha=\tfrac14(f_R^2+4f_Q f),\quad \beta=2\kappa^2 f_Q.
\label{eq:Psolution}
\end{equation}
Defining
\begin{equation}
\Sigma \equiv f_R \I + 2 f_Q \hat P = \frac{f_R}{2}\I + \sqrt{\alpha \I + \beta \hat T},
\label{eq:SigmaDef}
\end{equation}
the independent connection is Levi--Civita of an auxiliary metric $h_{\mu\nu}$ such that
\begin{equation}
h^{-1}=\frac{g^{-1}\Sigma}{\sqrt{\det\Sigma}},\qquad
R^\lambda{}_\mu(h)=\frac{1}{\sqrt{\det\Sigma}}\left(\frac{f}{2}\delta^\lambda_\mu+\kappa^2 T^\lambda{}_\mu\right).
\label{eq:h-relations}
\end{equation}

In this work we restrict to matter Lagrangians 
$\mathcal{L}_m(g,\psi)$ that do not depend explicitly on the independent
connection, so that the hypermomentum 
$\Delta_{\lambda}{}^{\mu\nu}\equiv -2\,\delta\mathcal{L}_m/\delta\Gamma^\lambda{}_{\mu\nu}$ 
vanishes. This sector includes all perfect fluids, scalars, electromagnetic fields, 
and also Dirac fields in the torsionless Palatini case, for which the spin connection 
reduces to the Levi--Civita one and no independent variation with respect to 
$\Gamma^\lambda{}_{\mu\nu}$ arises.

The absence of hypermomentum implies that the connection field equation is purely
algebraic and admits the exact matrix solution \eqref{eq:Psolution}--\eqref{eq:SigmaDef}. 
With non--vanishing hypermomentum the connection acquires new algebraic sources 
(spin and non-metricity currents), and the constitutive map 
$\Sigma(T)$ is replaced by a more general object $\Sigma(T,\Delta)$ with an enlarged 
operator basis. Extending the dual EFT developed here to 
hypermomentum--carrying matter is therefore conceptually straightforward---one 
would obtain additional local operators built from $\Delta_{\lambda}{}^{\mu\nu}$ 
and mixed $T$--$\Delta$ structures in the effective action---but working out this 
enlarged basis and matching it to explicit matter sectors (spinors with independent 
spin connection, nonminimally coupled gauge fields, etc.) lies beyond the scope 
of the present paper and is left for future work.

\section{Exact power–series form (matrix binomial)}\label{sec:exact-series}

A key result, already implicit in the exact matrix solution reported in~\cite{Olmo2011} and made explicit in~\cite{PerezTeruel2013, PerezTeruel2014}, is that the square--root matrix appearing in the Palatini map admits a closed binomial series in powers of the stress--energy tensor.

Starting from the exact expression
\begin{equation}
\hat P \;=\; -\frac{1}{4f_Q}\Big(f_R \hat I - 2\sqrt{\alpha \hat I + \beta \hat T}\Big),
\qquad
\Sigma \;\equiv\; f_R \hat I + 2 f_Q \hat P \;=\; \frac{f_R}{2}\hat I + \sqrt{\alpha \hat I + \beta \hat T},
\label{eq:SigmaExactAgain}
\end{equation}
with
\(
\alpha=\tfrac14(f_R^2+4f_Q f),\; \beta=2\kappa^2 f_Q,
\)
we factor out \(\sqrt{\alpha}\) and define \(X\equiv (\beta/\alpha)\hat T\). Using the generalized binomial theorem for matrices,
\begin{equation}
\boxed{\quad
\sqrt{\alpha \hat I + \beta \hat T}
\;=\; \sqrt{\alpha}\;\sum_{n=0}^{\infty} \binom{\tfrac12}{n}\, X^{n}
\;=\; \sqrt{\alpha}\;\sum_{n=0}^{\infty} \binom{\tfrac12}{n}\left(\frac{\beta}{\alpha}\right)^{\!n}\hat T^{\,n}\quad}
\label{eq:MatrixBinomialRoot}
\end{equation}
where the binomial coefficients are
\(
\displaystyle \binom{1/2}{n}=\frac{(1/2)_n}{n!}
=\frac{(1/2)(1/2-1)\cdots(1/2-n+1)}{n!}
=\frac{(-1)^{n-1}(2n-3)!!}{2^{n} n!}\;\; (n\ge1),
\)
with \(\binom{1/2}{0}=1\) and \((2n-3)!!\) the double factorial.

\paragraph*{Convergence.}
The series \eqref{eq:MatrixBinomialRoot} converges in any submultiplicative norm whenever
\begin{equation}
\|X\| \;=\; \left\| \frac{\beta}{\alpha}\hat T \right\| \;<\;1,
\qquad\text{equivalently}\qquad
\rho(X)=\max_i \left| \frac{\beta}{\alpha}\lambda_i(\hat T)\right| < 1,
\label{eq:convergence}
\end{equation}
with \(\{\lambda_i(\hat T)\}\) the eigenvalues of \(\hat T\). For physical stress tensors (diagonalizable over \(\mathbb{R}\)) this amounts to a bound on the matter scale \(|\beta T|<|\alpha|\), i.e.\ \(\rho\ll \rho_\star\) where \(\rho_\star\) is the characteristic density of the underlying theory (e.g.\ \(\rho_p=R_p/\kappa^2\) in quadratic models).\\

The convergence condition $\rho(X)<1$ in \eqref{eq:convergence} refers to 
the spectral radius of the mixed tensor $X^\mu{}_\nu=(\beta/\alpha)T^\mu{}_\nu$, 
i.e.\ the maximum modulus of its eigenvalues. 
In physical terms this is simply a statement that the matter variables 
$\rho$ and $p$ remain below the characteristic density scale of the theory, 
since for diagonalizable stress tensors (e.g.\ perfect fluids) the spectrum of 
$T^\mu{}_\nu$ is $\{-\rho,p,p,p\}$. 
Thus $\rho(X)<1$ is equivalent to 
\[
\left|\frac{\beta}{\alpha}\right|\max\{\rho,|p|\}<1,
\]
which provides a direct and transparent physical interpretation: the 
binomial/Neumann series converges whenever the matter density lies well below 
the intrinsic scale of the underlying Palatini model. No advanced spectral 
theory is required to apply this criterion.

From a global perspective, the condition $\rho(X)<1$ also delimits the domain 
where the \emph{principal} matrix root in \eqref{eq:MatrixBinomialRoot} is 
analytic and real. As one approaches the boundary $\rho(X)\to 1$ (for instance 
in ultra--dense regimes or for exotic equations of state with large pressures) 
the eigenvalues of $\alpha I+\beta T$ can approach the negative real axis, the 
binomial/Neumann expansion ceases to converge and the truncated dual EFT 
description breaks down. In that regime one must go back to the \emph{exact} 
algebraic relation \eqref{eq:Psolution} and track carefully the branch structure 
of the matrix square root. Whether the full non--perturbative solution describes 
a bounce, a regularisation of singularities or instead a strong--coupling regime 
depends on the specific $f(R,Q)$ model and on the matter sector; the present work 
focuses on the conservative EFT domain $\rho,|p|\ll\rho_\star$, where the 
principal branch is unambiguously selected and the power--series description is 
under perturbative control.

\paragraph*{Determinant and inverse as full series.}
Writing \(\Sigma=\frac{f_R}{2}\hat I+\sqrt{\alpha}\sum_{n\ge0}\binom{1/2}{n}X^n\),
it is convenient to factor \(\Sigma=A_0(\hat I+Y)\) with
\begin{equation}
A_0 \;\equiv\; \frac{f_R}{2}+\sqrt{\alpha},\qquad
Y \;\equiv\; \sum_{n=1}^{\infty} b_n X^n,\qquad
b_n \;\equiv\; \frac{\sqrt{\alpha}}{A_0}\binom{1/2}{n}.
\end{equation}
Then
\begin{equation}
\quad
\det\Sigma \;=\; A_0^{\,4}\;\exp\!\Big(\Tr \ln(\hat I+Y)\Big),
\qquad
\ln(\hat I+Y) \;=\; \sum_{k=1}^{\infty}\frac{(-1)^{k+1}}{k}\,Y^{k},\quad
\label{eq:detSeries}
\end{equation}
and
\begin{equation}
\quad
\Sigma^{-1} \;=\; \frac{1}{A_0}\,(\hat I+Y)^{-1}
\;=\; \frac{1}{A_0}\sum_{m=0}^{\infty}(-1)^{m}\,Y^{m},\quad
\label{eq:invSeries}
\end{equation}
both understood as absolutely convergent power series whenever \eqref{eq:convergence} holds. Since \(Y\) is a formal series in powers of \(\hat T\), all products \(Y^k\) generate finite linear combinations of \(\hat T^n\) at each order \(n\).

\paragraph*{Exact power–series for \(h_{\mu\nu}\).}
Using the defining relation
\(
h^{-1}=\dfrac{g^{-1}\Sigma}{\sqrt{\det\Sigma}}
\)
one finds the exact series representation
\begin{equation}
\quad
h_{\mu\nu}
\;=\; A_0\;\Big[\;\mathcal{D}^{-1/2}\!(Y)\;\Big]_{\mu}{}^{\alpha}\,
\Big[(\hat I+Y)^{-1}\Big]_{\alpha}{}^{\beta}\; g_{\beta\nu}\quad
\label{eq:hFullSeries}
\end{equation}
where
\(
\mathcal{D}^{-1/2}\!(Y) \equiv
\Big(\exp\big[\tfrac12\,\Tr \ln(\hat I+Y)\big]\Big)^{-1}\hat I
\)
acts as a scalar (times \(\hat I\)) that can be expanded with \eqref{eq:detSeries}. Combining \eqref{eq:detSeries} and \eqref{eq:invSeries} yields, to \emph{all} orders, a unique covariant series for \(h_{\mu\nu}\) in the tensor basis \(\{g_{\mu\nu},\,T_{\mu\nu},\,(T^2)_{\mu\nu},\ldots\}\).

Equations \eqref{eq:MatrixBinomialRoot}–\eqref{eq:hFullSeries} provide a compact, closed generating–function form of the Palatini map. In practice, for phenomenology one truncates these series at some finite order. In a next subsection we display the explicit expansion up to $\mathcal O(T^3)$. Now we give more details about the convergence of the matrix binomial expansion.
%------------------------------------------------------------
\subsection{Convergence of the matrix binomial expansion}
%------------------------------------------------------------

\begin{lemma}[Absolute convergence of the matrix binomial series]
\label{lemma:BinomialConvergence}
Let $X$ be a linear endomorphism on a finite-dimensional vector space and $\|\cdot\|$ any submultiplicative matrix norm.
If $\|X\|<1$, then the series
\[
\sum_{n=0}^{\infty}\binom{1/2}{n}\,X^n
\]
converges absolutely and defines the principal square root $(I+X)^{1/2}$. 
Moreover, convergence also holds whenever the spectral radius $\rho(X)<1$.
\end{lemma}
\begin{proof}[Sketch of proof]
Since $\sum_{n\ge0}\big|\binom{1/2}{n}\big|\,\|X\|^n$ converges for $\|X\|<1$ 
(it is dominated by a geometric series), the matrix series converges absolutely 
in any submultiplicative norm.  
Analytic functional calculus then implies that the sum coincides with the 
holomorphic function $f(X)$ with $f(z)=\sqrt{1+z}$ on the principal branch.
Since $\rho(X)<1$, one can always choose a matrix norm (equivalent to the usual 
operator norms) such that $\|X\|<1$; equivalently, all eigenvalues of $X$ lie 
strictly inside the open unit disk.  
This ensures absolute convergence of the binomial series.
\end{proof}
\begin{theorem}[Convergence for the Palatini map]\label{thm:PalatiniConvergence}
Let
\[
\Sigma \;=\; \frac{f_R}{2}\,I\;+\;\sqrt{\alpha I+\beta T}
\;=\; A_0\!\left(I+\sum_{n\ge1} b_n X^n\right),\qquad
X\equiv \frac{\beta}{\alpha}\,T,\quad 
b_n=\frac{\sqrt{\alpha}}{A_0}\binom{1/2}{n},
\]
where $\alpha=\tfrac14(f_R^2+4f_Q f)$, $\beta=2\kappa^2 f_Q$, 
$A_0=\frac{f_R}{2}+\sqrt{\alpha}$,
and $T^\mu{}_\nu$ is the mixed stress--energy tensor.
Fix the principal branch of the square root.  
If $\rho(X)<1$ (equivalently, $\max_i |\beta\,\tau_i/\alpha|<1$ for the eigenvalues $\{\tau_i\}$ of $T$), then:
\begin{enumerate}
\item The binomial series for $\sqrt{\alpha I+\beta T}$ converges absolutely to the principal root.
\item The series defining $\det\Sigma = A_0^{4}\exp[\Tr\ln(I+\sum_{n\ge1} b_n X^n)]$ 
and $\Sigma^{-1} = A_0^{-1}\sum_{m\ge0}(-1)^m(\sum_{n\ge1} b_n X^n)^m$ 
converge absolutely.
\item Consequently, the auxiliary metric  
\(
h_{\mu\nu}=g_{\mu\alpha}\Sigma^\alpha{}_\nu / \sqrt{\det\Sigma}
\)
admits a convergent expansion on the tensor basis 
$\{g_{\mu\nu},\,T_{\mu\nu},\,(T^2)_{\mu\nu},\ldots\}$.
\end{enumerate}
\end{theorem}

\begin{proof}
The result follows directly from Lemma~\ref{lemma:BinomialConvergence}.  
If $\rho(X)<1$, the Lemma ensures absolute convergence of the matrix binomial 
series 
\(\sum_{n\ge0}\binom{1/2}{n}X^n\),  
and therefore of 
\(Y=\sum_{n\ge1}b_n X^n\).
Absolute convergence of $Y$ implies 
$\|Y\|<1$ in the same domain, so the Neumann series 
\((I+Y)^{-1}=\sum_{m\ge0}(-1)^m Y^m\) converges absolutely, and
the trace--log series  
\(\Tr\ln(I+Y)=\sum_{k\ge1}(-1)^{k+1}\Tr(Y^k)/k\)  
also converges absolutely.
Since $\det\Sigma=A_0^{4}\exp[\Tr\ln(I+Y)]$ and 
$\Sigma^{-1}=A_0^{-1}(I+Y)^{-1}$ are products of absolutely convergent series,
their product structure in  
\(h_{\mu\nu}=g_{\mu\alpha}\Sigma^\alpha{}_\nu/\sqrt{\det\Sigma}\)  
also converges.  
\end{proof}
\begin{corollary}[Practical criterion for fluid sources]
If $T^\mu{}_\nu$ is diagonalizable with physical eigenvalues (for a perfect fluid $\{-\rho,p,p,p\}$), the convergence condition reduces to
\[
\max\!\left\{\left|\frac{\beta}{\alpha}\right|\,\rho,\ 
\left|\frac{\beta}{\alpha}\right|\,|p|\right\}<1.
\]
In quadratic models $f(R,Q)=R+(R^2+Q)/R_p$ this becomes
\(
\max\{\rho,|p|\}\ll \rho_p\equiv R_p/\kappa^2,
\)
and in EiBI gravity
\(
\max\{\rho,|p|\}\ll \rho_{\rm BI}\equiv 1/(\epsilon\kappa^2).
\)
\end{corollary}

\paragraph*{Remarks on branch choice and non-diagonalizable cases.}
(1) We adopt the principal branch of $\sqrt{\cdot}$, which is analytic on $\mathbb{C}\setminus(-\infty,0]$.
This requires the spectrum of $I+X$ to avoid the branch cut; the condition $\rho(X)<1$ is sufficient.
(2) If $T$ is not diagonalizable, the holomorphic functional calculus applies via the Jordan form:
the series still converges to the principal root provided the spectrum of $X$ lies inside the unit disk.
(3) In Lorentzian signature, $T^\mu{}_\nu$ need not be $g$-symmetric, but as an endomorphism it admits a Jordan decomposition; the spectral bounds above remain valid.
\paragraph*{Remark on the Schur method.}
The existence and construction of the principal matrix root 
$\sqrt{\alpha I+\beta T}$ follow from the Schur decomposition theorem:
any complex matrix admits a unitary decomposition 
$\alpha I+\beta T = Q\,U\,Q^\dagger$, 
with $U$ upper triangular and the eigenvalues of $\alpha I+\beta T$ on its diagonal (see e.g.~\cite{HighamFunctions,GolubVanLoan}).
The square root is then defined as
\begin{equation}
\sqrt{\alpha I+\beta T} = Q\,f(U)\,Q^\dagger,
\qquad f(U)^2 = U,
\end{equation}
where $f(U)$ is obtained recursively from the diagonal elements 
$f(\lambda_i)=\sqrt{\lambda_i}$ on the principal branch.
This procedure (Parlett recurrence) provides a constructive definition 
of the analytic root even when $T$ is not diagonalizable, 
and is numerically stable for any spectrum avoiding the negative real axis.
In this sense, the matrix solution reported in Ref.~\cite{PerezTeruel2013}
is well-defined for arbitrary matter sources satisfying the spectral condition
$\rho(X)<1$ of Theorem \ref{thm:PalatiniConvergence}.
\subsection{Algebraic expansion in the stress tensor}

From \eqref{eq:Psolution}--\eqref{eq:SigmaDef} we expand the matrix square root:
\begin{equation}
\sqrt{\alpha \I+\beta \hat T}=\sqrt{\alpha}\left(\I+\frac{\beta}{2\alpha}\hat T-\frac{\beta^2}{8\alpha^2}\hat T^2+\mathcal O(\hat T^3)\right).
\end{equation}
Thus
\begin{equation}
\Sigma = A_0\Big(\I + r\,\hat T + s\,\hat T^2+\mathcal O(\hat T^3)\Big),\qquad
A_0=\frac{f_R}{2}+\sqrt{\alpha},\quad r=\frac{\beta}{2A_0\sqrt{\alpha}},\quad s=-\frac{\beta^2}{8A_0\alpha^{3/2}}.
\label{eq:A0rs}
\end{equation}
Using $\det(\I+Y)=\exp[\Tr\ln(\I+Y)]$ and $(\I+Y)^{-1}=\I-Y+Y^2+\cdots$ with $Y=r\hat T+s\hat T^2$, we obtain (up to $\mathcal O(T^2)$)
\begin{align}
\sqrt{\det\Sigma}&=A_0^2\!\left[1+\frac{r}{2}\Tr T+\Big(\frac{s}{2}-\frac{r^2}{4}\Big)\Tr(T^2)+\frac{r^2}{8}(\Tr T)^2\right],
\label{eq:detSigma}\\
\Sigma^{-1}&=\frac{1}{A_0}\Big[\I - r\,\hat T + (r^2-s)\,\hat T^2\Big].
\label{eq:SigmaInv}
\end{align}
Since $h=g\,\Sigma/\sqrt{\det\Sigma}$ (equivalent to \eqref{eq:h-relations}), lowering indices with $g$:
\begin{align}
h_{\mu\nu}
&=A_0\Big\{g_{\mu\nu}
+ r\big[\tfrac12(\Tr T)\,g_{\mu\nu}-T_{\mu\nu}\big]
+ \mathcal H^{(2)}_{\mu\nu}\Big\}+\mathcal O(T^3),
\label{eq:h-expansion}\\
\mathcal H^{(2)}_{\mu\nu}
&=(r^2-s)(T^2)_{\mu\nu}
-\tfrac12 r^2(\Tr T)\,T_{\mu\nu}
+\Big[\tfrac12 s\,\Tr(T^2)+\tfrac18 r^2\big((\Tr T)^2-2\,\Tr(T^2)\big)\Big] g_{\mu\nu}.
\end{align}

% ------------------------------------------------------------
\section{Dual EFT description}
% ------------------------------------------------------------
Inserting \eqref{eq:detSigma}--\eqref{eq:h-expansion} into \eqref{eq:h-relations} and expressing the equations in the physical metric yields an \emph{effective Einstein's field equations}
\begin{equation}
G_{\mu\nu}(g)=\kappa^2\Big[T_{\mu\nu}+(\Delta T)_{\mu\nu}\Big]+\mathcal O(T^3),
\label{eq:EinsteinEff}
\end{equation}
with a local, algebraic correction
\begin{align}
(\Delta T)_{\mu\nu} &=
c_1\,(\Tr T)\,g_{\mu\nu}
+ c_2\,T_{\mu\nu}
+ c_3\,(T^2)_{\mu\nu}
+ c_4\,(\Tr T)\,T_{\mu\nu}
+ c_5\,\Tr(T^2)\,g_{\mu\nu}
+ c_6\,(\Tr T)^2 g_{\mu\nu},
\label{eq:DeltaTbasis}
\end{align}
whose coefficients are fixed functions of $(A_0,r,s)$ (and of $f$ at low order). At leading orders one finds schematically
\begin{equation}
c_1=\frac{r}{2}+\mathcal O(r^2,s),\qquad
c_2=-\,r+\mathcal O(r^2,s),\qquad
c_{3,4,5,6}=\mathcal O(r^2,s).
\end{equation}
Equation \eqref{eq:EinsteinEff} is equivalent to the Euler--Lagrange equations of the local effective action
\begin{equation}
S_{\rm eff}[g,\psi]=\int d^4x\,\sqrt{-g}\left[
\frac{1}{2\kappa^2}R(g)-\Lambda_{\rm eff}
+\mathcal L_m(g,\psi)
+\frac{a_1}{\rho_\star}\,T
+\frac{a_2}{\rho_\star^{2}}\,T^2
+\frac{a_3}{\rho_\star^{2}}\,T_{\mu\nu}T^{\mu\nu}
+\ldots\right],
\label{eq:Seff}
\end{equation}
with $\{a_i\}$ linearly related to $\{c_i\}$ and a characteristic density scale $\rho_\star$ set by the underlying theory (e.g. $\rho_\star=\rho_p\equiv R_p/\kappa^2$ in quadratic models). This provides a \emph{dual EFT in the stress--energy sector}, complementary to curvature--based EFTs. The structure of the expansion and the way Palatini and EiBI models populate the $T$--tower are illustrated schematically in Fig.~\ref{fig:eft-structure}.

\paragraph*{Notation.}
In the effective action, the quadratic invariants built from the 
stress--energy tensor are understood as
\[
T^2 \equiv (g^{\mu\nu}T_{\mu\nu})^2,
\qquad
T_{\mu\nu}T^{\mu\nu} \equiv T^\mu{}_\nu T^\nu{}_\mu = \Tr(T^2),
\]
which are algebraically independent and correspond to the 
two possible scalar contractions of $T_{\mu\nu}$ at second order.
Their role mirrors the pair $(R^2,\,R_{\mu\nu}R^{\mu\nu})$ in 
quadratic gravity.

\subsection{Illustrative model $f(R,Q)=R+(R^2+Q)/R_p$}
% ------------------------------------------------------------
For $f_R=1+2R/R_p$, $f_Q=1/R_p$ and $R=-\kappa^2 T$, one finds to leading order
\begin{equation}
A_0\simeq 1-\frac{3\kappa^2 T}{R_p},\qquad
r\simeq \frac{2\kappa^2}{R_p}=\frac{2}{\rho_p},\qquad
s\simeq -\frac{4\kappa^4}{R_p^2}=-\frac{4}{\rho_p^2},\qquad
\rho_p\equiv \frac{R_p}{\kappa^2}.
\end{equation}
Hence $c_{1,2}=\mathcal O(1/\rho_p)$ and $c_{3\text{--}6}=\mathcal O(1/\rho_p^2)$, so deviations from GR are organized by $\epsilon=\rho/\rho_p$ and the series converges rapidly for $\rho\ll\rho_p$.

\section{Extension to Eddington--inspired Born--Infeld (EiBI)}
% ------------------------------------------------------------
The EiBI theory can be written as a determinantal action whose variation yields the algebraic relation
\begin{equation}
\sqrt{-q}\, q^{\mu\nu} = \sqrt{-g}\,\Big(\lambda\, g^{\mu\nu} - \epsilon\,\kappa^2\, T^{\mu\nu}\Big),
\label{eq:EiBI_master}
\end{equation}
with constants $\lambda$ (dimensionless) and $\epsilon$ (length$^2$). Defining $q_{\mu\nu}=g_{\mu\alpha}\Sigma^\alpha{}_\nu$
(one has $q=g\,\det\Sigma$ and $q^{\mu\nu}=(\Sigma^{-1})^\mu{}_\alpha g^{\alpha\nu}$), Eq.~\eqref{eq:EiBI_master} implies
\begin{equation}
\sqrt{\det\Sigma}\;\Sigma^{-1}= \lambda\,\I - \epsilon\,\kappa^2\,\hat T \;\equiv\; A,
\qquad\Rightarrow\qquad
\boxed{\;\Sigma=(\det A)^{1/2}\,A^{-1}\;},\qquad
\sqrt{-q}=\sqrt{-g}\,\sqrt{\det\Sigma}=\sqrt{-g}\,\sqrt{\det A}.
\label{eq:Sigma-EiBI-exact}
\end{equation}

It is convenient to set
\begin{equation}
\alpha_{\rm BI}\equiv \frac{\epsilon\kappa^2}{\lambda},\qquad
A=\lambda\big(\I-\alpha_{\rm BI}\,\hat T\big).
\end{equation}
Then
\begin{equation}
A^{-1}=\frac{1}{\lambda}\sum_{n=0}^{\infty}\alpha_{\rm BI}^{\,n}\hat T^{\,n},
\qquad
(\det A)^{1/2}=\lambda^{2}\exp\!\left[\tfrac12\,\Tr\ln\!\big(\I-\alpha_{\rm BI}\hat T\big)\right].
\label{eq:EiBI-Neumann+det}
\end{equation}
Combining \eqref{eq:EiBI-Neumann+det} gives the exact analytic form
\begin{equation}
\Sigma
=\lambda\,\exp\!\left[\tfrac12\,\Tr\ln\!\big(\I-\alpha_{\rm BI}\hat T\big)\right]
\sum_{n=0}^{\infty}\alpha_{\rm BI}^{\,n}\hat T^{\,n}.
\label{eq:Sigma-EiBI-analytic}
\end{equation}

\paragraph*{Low-order expansion.}
Expanding \eqref{eq:Sigma-EiBI-analytic} up to ${\cal O}(T^2)$ one finds
\begin{align}
\Sigma &= \lambda\Bigg[
\I
+ \alpha_{\rm BI}\,\hat T
- \frac{\alpha_{\rm BI}}{2}\,\Tr(\hat T)\,\I
+ \alpha_{\rm BI}^2\Big(
\hat T^{2}
-\frac{1}{2}\Tr(\hat T)\,\hat T
-\frac{1}{4}\Tr(\hat T^{2})\,\I
+\frac{1}{8}\Tr(\hat T)^{2}\,\I
\Big)
+{\cal O}(T^3)\Bigg], \label{eq:Sigma-exp}\\[2mm]
\sqrt{\det\Sigma}&=\sqrt{\det A}
=\lambda^{2}\Bigg[1
-\frac{\alpha_{\rm BI}}{2}\Tr(\hat T)
-\frac{\alpha_{\rm BI}^{2}}{4}\Tr(\hat T^{2})
+\frac{\alpha_{\rm BI}^{2}}{8}\Tr(\hat T)^{2}
+{\cal O}(T^3)\Bigg].
\label{eq:detSigma-exp}
\end{align}
In particular, the linear response is controlled by $+\alpha_{\rm BI}$, while the first trace counterterm comes with $-\alpha_{\rm BI}/2$.
\paragraph*{Dual EFT interpretation.}
Equations \eqref{eq:Sigma-exp}–\eqref{eq:detSigma-exp} reproduce the universal
$h$–expansion in the operator basis $\{ \hat I,\, \hat T,\, \hat T^2,\, \mathrm{Tr}(\widehat T)\hat I,\ldots\}$.
Therefore,
\begin{equation}
G_{\mu\nu}(g)=\kappa^2\Big[T_{\mu\nu}+(\Delta T)^{\rm (BI)}_{\mu\nu}\Big]+{\cal O}(T^3),
\end{equation}
with $(\Delta T)^{\rm (BI)}_{\mu\nu}$ expanded on the same basis as in \eqref{eq:DeltaTbasis}. The leading scalings are
\begin{equation}
c^{\rm (BI)}_{1}\sim +\alpha_{\rm BI},\qquad
c^{\rm (BI)}_{\rm tr}\sim -\frac{\alpha_{\rm BI}}{2},\qquad
c^{\rm (BI)}_{2,3,4,5}\sim {\cal O}(\alpha_{\rm BI}^{2}),
\end{equation}
where $c^{\rm (BI)}_{1}$ multiplies $T_{\mu\nu}$, $c^{\rm (BI)}_{\rm tr}$ multiplies $T\,g_{\mu\nu}$,
and the remaining coefficients are quadratic in $\alpha_{\rm BI}$. The associated dual action keeps the form
\begin{equation}
S_{\rm eff}^{\rm (BI)}[g,\psi]=\int d^4x\,\sqrt{-g}\left[
\frac{R}{2\kappa^2}-\Lambda_{\rm eff}
+\mathcal L_m
+\frac{b_1}{\rho_{\rm BI}}\,T
+\frac{b_2}{\rho_{\rm BI}^{2}}\,T^2
+\frac{b_3}{\rho_{\rm BI}^{2}}\,T_{\mu\nu}T^{\mu\nu}
+\ldots\right],\qquad
\rho_{\rm BI}\equiv \frac{1}{\epsilon\kappa^2},
\end{equation}
with updated numerical coefficients $b_i=b_i(\alpha_{\rm BI})$ obtained from \eqref{eq:Sigma-exp}–\eqref{eq:detSigma-exp}. EiBI thus realizes a determinantal resummation of the stress--energy tower depicted in Fig.~\ref{fig:eft-structure}.
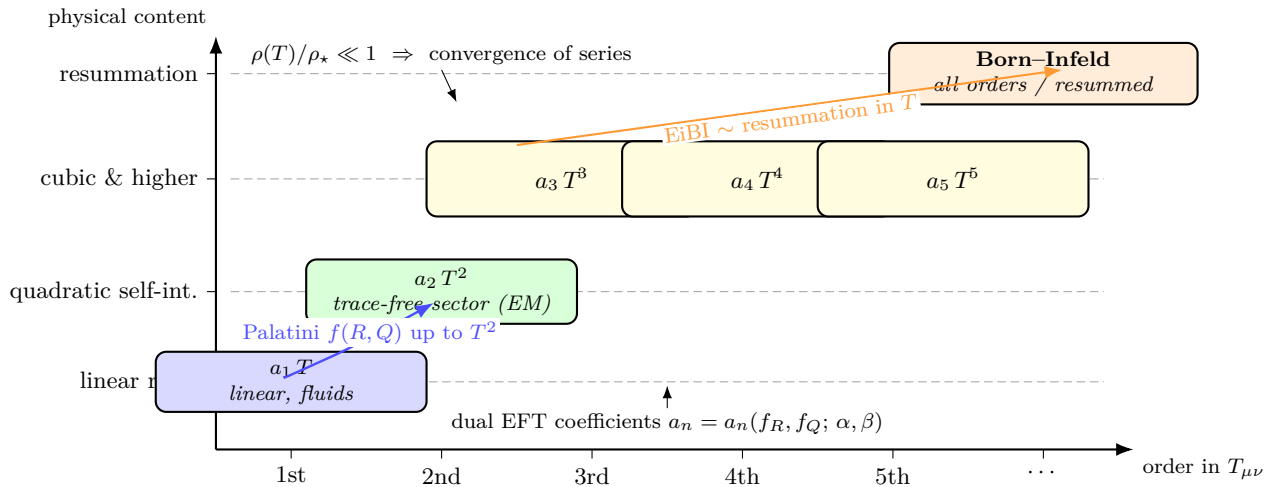
\begin{figure}[t]
\centering
\begin{tikzpicture}[
  >=Latex,
  axis/.style={-Latex, thick},
  tick/.style={thin},
  box/.style={draw, rounded corners, thick, align=center, font=\footnotesize, fill=#1!15,
              minimum width=3.6cm, minimum height=8mm},
  ybox/.style={draw, rounded corners, thick, align=center, font=\small, fill=yellow!15,
               minimum width=3.6cm, minimum height=10mm},
  label/.style={font=\footnotesize}
]

% Axes
\draw[axis] (0,0) -- (12.2,0) node[below right,label] {order in $T_{\mu\nu}$};
\draw[axis] (0,0) -- (0,5.5) node[above left,label] {physical content};

% X ticks
\foreach \x/\txt in {1/1st,3/2nd,5/3rd,7/4th,9/5th,11/$\cdots$}
  \draw[tick] (\x,0) -- (\x,-0.12) node[below] {\txt};

% Y labels
\foreach \y/\txt in {0.9/{linear resp.},2.1/{quadratic self-int.},
                     3.6/{cubic \& higher},5.0/{resummation}}
  \draw[tick] (0,\y) -- (-0.12,\y) node[left] {\txt};

% dashed guide lines (behind boxes)
\foreach \y in {0.9,2.1,3.6,5.0}
  \draw[densely dashed, gray!70, line cap=round] (0.2,\y) -- (11.8,\y);

% Boxes
\node[box=blue]   at (1,0.9)  {$a_1\,T$\\\textit{linear, fluids}};
\node[box=green]  at (3,2.1)  {$a_2\,T^2$\\\textit{trace-free sector (EM)}};

% Yellow boxes (normal font; clear exponents)
\node[ybox] at (4.6,3.6)  {$a_3\,T^{\!3}$};
\node[ybox] at (7.2,3.6)  {$a_4\,T^{\!4}$};
\node[ybox] at (9.8,3.6)  {$a_5\,T^{\!5}$};

% Born–Infeld
\node[box=orange, minimum width=4.1cm] at (11,5.0)
  {\textbf{Born--Infeld}\\\textit{all orders / resummed}};

% --- Palatini arrow (robust: separate path + label node) ---
\draw[blue!70,->,line width=0.9pt,fill=none]
  (0.9,0.95) .. controls (1.8,1.35) .. (2.9,1.95);
\node[blue!70,fill=white,inner sep=1pt,outer sep=0pt,font=\footnotesize]
  at (2.05,1.55) {Palatini $f(R,Q)$ up to $T^2$};

% EiBI arrow
\draw[thick,orange!80,->,fill=none] (4.0,4.05) -- (11.25,5.05)
  node[midway,below,sloped,fill=white,inner sep=1pt]
  {\footnotesize EiBI $\sim$ resummation in $T$};

% Annotation: coefficients
\node[label,align=left] at (6.0,0.35)
 {dual EFT coefficients $a_n = a_n(f_R,f_Q;\,\alpha,\beta)$};
\draw[->,thin] (6.0,0.55) -- (6.0,0.85);

% Convergence note
\node[label,align=left] at (3.0,5.25)
 {$\rho(T)/\rho_\star \ll 1$ \ $\Rightarrow$ \ convergence of series};
\draw[->,thin] (3.0,5.05) -- (3.2,4.62);

\end{tikzpicture}
\caption{
\textbf{Structure of the dual EFT expansion.}
The analytic map $\Sigma(T)$ induces a local effective action 
$S_{\rm eff}[g,T]=S_{\rm GR}+\sum_{n\ge1} a_n\,T^n/\rho_\star^{\,n-1}$.
For perfect-fluid sources, the linear term ($\propto T$) governs leading departures from GR, 
while trace-free sectors (e.g.\ electromagnetism) start at quadratic order ($\propto T^2$).
Quadratic Palatini models $f(R,Q)$ effectively populate the first two orders, 
whereas Eddington-inspired Born--Infeld corresponds to a resummation to all orders in $T$.
Convergence holds in the controlled regime $\rho(T)/\rho_\star\ll1$ (see App.~\ref{app:rigorous-convergence}).
}
\label{fig:eft-structure}
\end{figure}

\paragraph*{Convergence of the series.}
The analytic form~\eqref{eq:Sigma-EiBI-analytic} involves two factors:
an exponential of a trace logarithm, $\exp[\frac12\Tr\ln(\I-\alpha_{\rm BI}\hat T)]$, and the
Neumann series $(\I-\alpha_{\rm BI}\hat T)^{-1}=\sum_{n=0}^{\infty}(\alpha_{\rm BI}\hat T)^{n}$.
The scalar prefactor is analytic as long as $(\I-\alpha_{\rm BI}\hat T)$ is invertible, while
the Neumann series converges if and only if the spectral radius of $\alpha_{\rm BI}\hat T$ satisfies
\begin{equation}
\rho(\alpha_{\rm BI}\hat T)<1.
\end{equation}
For a perfect fluid, $\hat T^\mu{}_\nu={\rm diag}[-\rho,p,p,p]$, this condition reduces to
$|\alpha_{\rm BI}\rho|<1$ and $|\alpha_{\rm BI}p|<1$.
Defining the Born--Infeld density scale $\rho_{\rm BI}\equiv(\epsilon\kappa^2)^{-1}$,
the convergence domain is therefore
\begin{equation}
\boxed{\;\rho,\,|p| \ll \rho_{\rm BI}\;}
\end{equation}
which coincides with the binomial convergence criterion discussed in Sec.~\ref{sec:exact-series}.
Thus, the determinantal (EiBI) resummation preserves the same physical regime of validity
as the Palatini expansion in powers of $T_{\mu\nu}$, ensuring that both formulations are
consistent within the low-density effective field theory domain.
For arbitrary matter sources, the Neumann expansion
$(\I-\alpha_{\rm BI}\hat T)^{-1}=\sum_{n=0}^{\infty}(\alpha_{\rm BI}\hat T)^{n}$
converges if and only if the spectral radius $\rho(\alpha_{\rm BI}\hat T)<1$.
This holds independently of the diagonalizability of $\hat T$.
In the perfect--fluid case, this reduces to
$|\alpha_{\rm BI}\rho|,|\alpha_{\rm BI}p|<1$;
for electromagnetic fields to $|\alpha_{\rm BI}\mathcal E|<1$;
and for null or type--N matter ($\hat T^2=0$) the series truncates exactly.
Hence, the Born--Infeld resummation remains valid for all matter types
within the same low--energy domain
\(\|\,\alpha_{\rm BI}\hat T\,\|<1\),
which corresponds physically to densities and field strengths
well below the Born--Infeld scale $\rho_{\rm BI}=1/(\epsilon\kappa^2)$.

\section{Some examples}
To illustrate how the dual EFT map operates in concrete situations, we now apply the 
formalism to representative matter sources. The goal is not to develop full 
phenomenology, but to make explicit how the analytic structures derived in the previous 
sections determine the auxiliary metric $h_{\mu\nu}$ and the effective observables 
constructed from it. The general flow from an input stress tensor $T^\mu{}_\nu$ to the 
corresponding effective quantities is summarized schematically in Fig.~3.
\subsection{Perfect fluid (cosmology \& stars)}
For a perfect fluid $T^\mu{}_\nu=\mathrm{diag}(-\rho,p,p,p)$ one has
$\Tr T=-\rho+3p$, $\Tr(T^2)=\rho^2+3p^2$ and $(T^2)^\mu{}_\nu=\mathrm{diag}(\rho^2,p^2,p^2,p^2)$.
Using Eq.~\eqref{eq:h-expansion}, up to ${\cal O}(T)$,
\begin{align}
h_{tt}&=A_0\,g_{tt}\Big[1+\tfrac{r}{2}(\rho+3p)\Big],\qquad
h_{rr}=A_0\,g_{rr}\Big[1+\tfrac{r}{2}(-\rho+p)\Big],\nonumber\\
h_{\theta\theta}&=A_0\,g_{\theta\theta}\Big[1+\tfrac{r}{2}(-\rho+p)\Big],\qquad
h_{\phi\phi}=\sin^2\theta\,h_{\theta\theta}.
\end{align}
The Einstein--like equations in $h_{\mu\nu}$ map to an effective Einstein theory in $g_{\mu\nu}$ with
\begin{equation}
\rho_{\rm eff}=\frac{1}{\sqrt{\det\Sigma}}\!\left(\rho+\frac{f}{2\kappa^2}\right),
\qquad
p_{\rm eff}=\frac{1}{\sqrt{\det\Sigma}}\!\left(p-\frac{f}{2\kappa^2}\right),
\label{eq:rhopeff}
\end{equation}
and, to first order in $T$,
\begin{equation}
\rho_{\rm eff}\approx \rho\Big[1-\tfrac12 r(-\rho+3p)\Big]+\frac{f}{2\kappa^2A_0},\qquad
p_{\rm eff}\approx p\Big[1-\tfrac12 r(-\rho+3p)\Big]-\frac{f}{2\kappa^2A_0}.
\label{eq:rhopeff-LO}
\end{equation}
These expressions represent the first terms in the universal dual EFT expansion and 
apply independently of the underlying gravitational Lagrangian, provided the spectral 
condition $\rho(X)<1$ holds.
\paragraph*{Cosmology (FLRW).}
For $g_{\mu\nu}$ FLRW and barotropic $p=w\rho$, Eq.~\eqref{eq:rhopeff} yields a modified Friedmann equation
\begin{equation}
3H_h^2=\kappa^2\,\rho_{\rm eff}+\Lambda_{\rm eff},\qquad
\Lambda_{\rm eff}\equiv \frac{f}{2A_0}\,,
\end{equation}
with $H_h$ the Hubble rate of $h_{\mu\nu}$; using the algebraic map $h\leftrightarrow g$ one obtains $H$ in the physical frame. At leading order the departure from GR is controlled by $r\propto 1/\rho_\star$ and is $\mathcal O(\rho/\rho_\star)$. A schematic comparison between the GR, Palatini-EFT and Born--Infeld–resummed behaviours of $H^2(\rho)$ is shown in Fig.~\ref{fig:cosmo-schematic}.
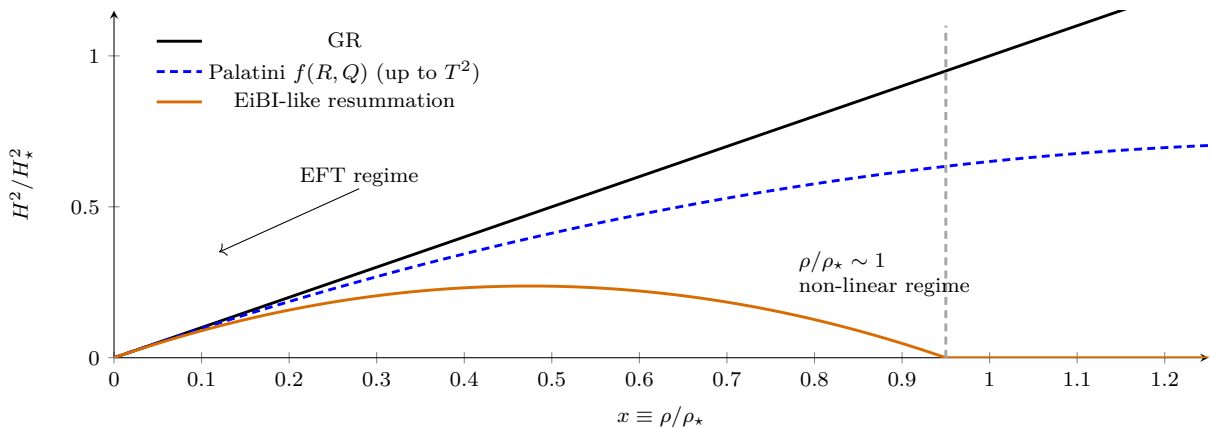
\begin{figure}[t]
\centering
\begin{tikzpicture}
\begin{axis}[
    width=0.9\linewidth, height=6.2cm,
    xmin=0, xmax=1.25,
    ymin=0, ymax=1.15,
    axis lines=left,
    xlabel={$x \equiv \rho/\rho_\star$},
    ylabel={$H^2/H_\star^2$},
    label style={font=\footnotesize},
    ticklabel style={font=\footnotesize},
    legend style={font=\footnotesize, draw=none, fill=none, at={(0.03,0.97)}, anchor=north west},
    every axis plot/.style={line width=1.1pt},
    domain=0:1.25,
    samples=400,
    clip=true,
    grid=none
]

% ----- illustrative parameters -----
\def\aTwo{-0.35}   % quadratic coefficient in truncated Palatini example
\def\xb{0.95}      % x at which resummed curve bounces (y=0)

% ----- GR: y = x -----
\addplot[black] {x};
\addlegendentry{GR}

% ----- Palatini truncation to O(T^2): y = x + a2 x^2 (cut at y>=0) -----
\addplot[blue, densely dashed, domain=0:1.25] 
  ({x}, {max(0, x + \aTwo*x^2)});
\addlegendentry{Palatini $f(R,Q)$ (up to $T^2$)}

% ----- EiBI-like resummation: y = x*(1 - x/xb) (schematic) -----
\addplot[orange!85!black] 
  ({x}, {max(0, x*(1 - x/\xb))});
\addlegendentry{EiBI-like resummation}

% ----- bounce guide (x = xb) -----
\addplot[densely dashed, gray!70] coordinates {(\xb,0) (\xb,1.1)};
\node[gray!70, anchor=north east, font=\footnotesize] at (axis cs:\xb,0) {$x_b$};

% ----- annotations -----
\node[font=\footnotesize] at (axis cs:0.28,0.60) {EFT regime};
\draw[->,thin] (axis cs:0.28,0.56) -- (axis cs:0.12,0.35);

\node[font=\footnotesize, align=left] at (axis cs:0.88,0.28)
{\(\rho/\rho_\star\sim 1\) \\ non-linear regime};

\end{axis}
\end{tikzpicture}
\caption{
\textbf{Cosmological application (schematic).}
Effective Hubble rate versus normalized density \(x=\rho/\rho_\star\).
General Relativity gives \(H^2/H_\star^2=x\) (black).
The dual EFT truncation up to \(T^2\) (blue, dashed) is shown with a representative coefficient \(a_2\).
A Born--Infeld--like resummation (orange) captures non-linear saturation and a bounce at \(x=x_b\).
Parameters \((a_2,x_b)\) are illustrative for visualization; the formalism fixes their mapping to \((f_R,f_Q;\alpha,\beta)\).
}
\label{fig:cosmo-schematic}
\end{figure}

\paragraph*{Stars (static, spherically symmetric).}
In the stellar case, \eqref{eq:rhopeff-LO} feeds the TOV structure in the $h$-frame; the mapping $r_h^2=h_{\theta\theta}$ gives the small geometric rescalings needed to obtain $M(R)$ shifts. The net effect at first order is equivalent to a mild renormalization of $(\rho,p)$ plus tiny anisotropy-like geometric factors from $h_{rr},h_{\theta\theta}$.

\paragraph*{Advantages and physical insight.}
The cosmological sector illustrates particularly well the advantages of the dual EFT construction. 
In standard treatments of Palatini $f(R,Q)$ or EiBI models, the modified Friedmann equations are obtained 
by explicitly solving the algebraic relation between $q_{\mu\nu}$ and $g_{\mu\nu}$ for a perfect-fluid source, 
which typically requires a case-by-case numerical inversion. 
In contrast, the present formulation provides an \emph{analytic and universal} expansion 
for the effective energy density and pressure, Eqs.~\eqref{eq:rhopeff}--\eqref{eq:rhopeff-LO}, 
valid for any barotropic equation of state and for the entire $f(R,Q)$ or Born--Infeld class. 

From a physical standpoint, this expansion has three immediate benefits:
\begin{enumerate}
\item[(i)] It organizes deviations from GR by a small, dimensionless parameter 
$\epsilon \equiv \rho/\rho_\star$, making the regime of validity explicit and allowing 
a direct identification of the leading corrections as analytic, local functions of the matter variables.
\item[(ii)] It allows one to read off the modified Friedmann dynamics 
without integrating differential equations: the expansion translates geometric nonlinearities 
into an effective ``stiffness'' of the fluid, $\rho_{\rm eff}(\rho,p)$, 
that can be implemented directly in cosmological codes or phenomenological models.
\item[(iii)] The same algebraic structure applies to any matter source. 
In the case of trace-free sources (e.g.\ radiation or electromagnetic fields), the leading corrections vanish 
at $\mathcal{O}(T)$ and only appear at quadratic order, revealing a built-in suppression mechanism 
for relativistic components that is not manifest in the usual formulations.
\end{enumerate}
As a consequence, the framework provides a transparent EFT-like hierarchy for the matter sector, 
bridging high-density cosmology and the physics of compact objects within a single algebraic expansion. 
In particular, the bounce or avoidance of singularities discussed in earlier Palatini and EiBI models 
appears here as a controlled resummation of the series in $\rho/\rho_\star$, 
which can be truncated or extended according to the physical density range of interest. 
This qualitative behaviour is captured in the schematic plot of Fig.~\ref{fig:cosmo-schematic}.

\subsection{Electromagnetic field (trace-free source)}
For Maxwell electrodynamics in four dimensions, $T^\mu{}_\mu=0$. Then
the ${\cal O}(T)$ deformation of $h_{\mu\nu}$ simplifies to
\begin{equation}
h_{\mu\nu}=A_0\Big[g_{\mu\nu}-r\,T_{\mu\nu}\Big]+\mathcal O(T^2),
\qquad
\sqrt{\det\Sigma}=A_0^2\Big[1-\tfrac18 r^2\,\Tr(T^2)\Big]+\mathcal O(T^3),
\end{equation}
so the leading correction is purely proportional to $T_{\mu\nu}$ and the first trace contribution enters only at ${\cal O}(T^2)$ through $\Tr(T^2)$. This yields a particularly clean laboratory for lensing or black-hole exteriors with electromagnetic hair, since the power counting depends on a single small parameter $r$ (or $r_{\rm BI}$ in EiBI).

\subsection{Compact stars: TOV at first order}
In the $h$-frame the equilibrium equations retain the standard form
\begin{equation}
\frac{dm_h}{dr_h}=4\pi r_h^2\,\rho_{\rm eff},\qquad
\frac{dp_{\rm eff}}{dr_h}
=-\frac{(\rho_{\rm eff}+p_{\rm eff})\,(m_h+4\pi r_h^3 p_{\rm eff})}{r_h(r_h-2G m_h)}\,,
\end{equation}
with $(\rho_{\rm eff},p_{\rm eff})$ from \eqref{eq:rhopeff}. Using \eqref{eq:rhopeff-LO} one immediately gets the leading shifts
\begin{equation}
\delta\rho \equiv \rho_{\rm eff}-\rho \simeq -\frac{r}{2}\,(-\rho+3p)\,\rho+\frac{f}{2\kappa^2A_0},\quad
\delta p \equiv p_{\rm eff}-p \simeq -\frac{r}{2}\,(-\rho+3p)\,p-\frac{f}{2\kappa^2A_0},
\end{equation}
which translate into controlled deviations in $M(R)$ and tidal deformabilities at ${\cal O}(\rho/\rho_\star)$ (or ${\cal O}(\epsilon\kappa^2\rho)$ in EiBI). This ``plug-and-play'' recipe makes the phenomenology straightforward once an EoS is specified.
\begin{figure}[t]
\centering
\begin{tikzpicture}[
  >=Latex,
  node distance=7mm and 10mm,
  box/.style={
    draw,
    rounded corners,
    thick,
    minimum height=8mm,
    text width=4.2cm,
    align=center,
    font=\footnotesize,
    fill=#1!15,
    preaction={fill=black,opacity=0.05,transform canvas={xshift=1pt,yshift=-1pt}}
  },
  arrow/.style={-{Latex[length=2.5mm]}, thick}
  ]

% Nodes
\node[box=blue]   (T)     {$T_{\mu\nu}$\\[1pt]\textit{(matter source)}};
\node[box=green,  right=of T] (Sigma)  {$\displaystyle \Sigma=\frac{f_R}{2}I+\sqrt{\alpha I+\beta T}$\\[1pt]\textit{(matrix map)}}; 
\node[box=yellow, right=of Sigma] (h)  {$\displaystyle h_{\mu\nu}=\frac{g_{\mu\alpha}\Sigma^\alpha{}_\nu}{\sqrt{\det\Sigma}}$\\[1pt]\textit{(auxiliary metric)}};

\node[box=orange, below=of Sigma] (eff) {$\displaystyle S_{\rm eff}[g,T]=S_{\rm GR}+\sum a_n\,\frac{T^n}{\rho_\star^{\,n-1}}$\\[1pt]\textit{(dual EFT action)}};
\node[box=red,    right=of eff]   (obs) {\textit{Observables:}\\[1pt]Cosmology / Stars / EM};

% Arrows
\draw[arrow] (T) -- (Sigma);
\draw[arrow] (Sigma) -- (h);
\draw[arrow] (Sigma) -- (eff);
\draw[arrow] (h) -- (eff);
\draw[arrow] (eff) -- (obs);

\end{tikzpicture}

\caption{
Schematic flow of the dual EFT construction.
Starting from a matter source $T_{\mu\nu}$, the algebraic matrix map $\Sigma$ encodes the nonlinear gravitational response of the connection.
From $\Sigma$ one obtains the auxiliary metric $h_{\mu\nu}$ satisfying Einstein--like equations.
The relation between $h_{\mu\nu}$ and the physical metric $g_{\mu\nu}$ defines an effective local action $S_{\rm eff}[g,T]$
expanded in powers of $T_{\mu\nu}/\rho_\star$.
This unified algebraic framework applies to cosmological fluids, stellar interiors, and electromagnetic fields.
}
\label{fig:flowEFT}
\end{figure}
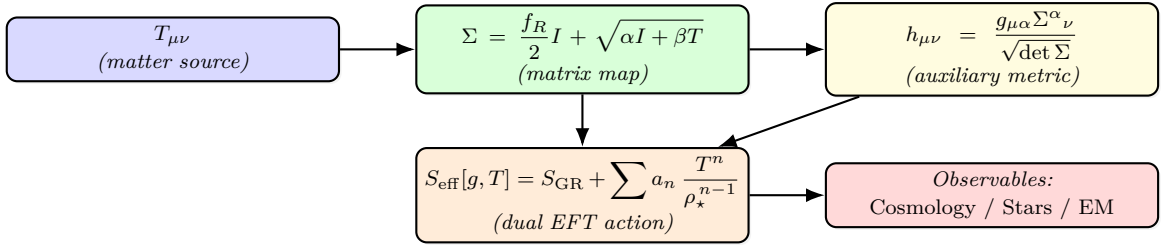
\section{Symmetric teleparallel $f(Q)$ gravity in the dual analytic EFT framework}
\label{sec:dual-fQ-final}

Symmetric teleparallel gravity is defined by imposing vanishing curvature and torsion
on the affine connection while allowing for non--metricity.  The independent variables
are $(g_{\mu\nu},\Gamma^\alpha{}_{\mu\nu})$ subject to
\begin{equation}
R^\alpha{}_{\beta\mu\nu}(\Gamma)=0, 
\qquad
T^\alpha{}_{\mu\nu}(\Gamma)=0, 
\qquad
Q_{\alpha\mu\nu}\equiv\nabla_\alpha g_{\mu\nu}\neq0 .
\end{equation}
The non--metricity scalar in the STEGR convention is
\begin{equation}
Q = - g^{\mu\nu}\!\left(L^\alpha{}_{\mu\beta} L^\beta{}_{\nu\alpha}
      - L^\alpha{}_{\mu\nu} L^\beta{}_{\alpha\beta}\right),
\qquad
L^\alpha{}_{\mu\nu}
=\tfrac12 g^{\alpha\lambda}
\!\left(\nabla_\mu g_{\lambda\nu}+\nabla_\nu g_{\lambda\mu}
      -\nabla_\lambda g_{\mu\nu}\right),
\label{eq:Q-def}
\end{equation}
and the action
\begin{equation}
S[g,\Gamma,\psi]=\int d^4x\,\sqrt{-g}\,f(Q)+S_m[g,\psi]
\label{eq:fQ-action}
\end{equation}
yields second--order field equations for the metric.  
General Relativity is recovered for $f(Q)=Q$, while nonlinear choices describe
viable modified gravities in cosmology and astrophysics.

%------------------------------------------------------------
\subsection{Motivation from the dual matter expansion}
%------------------------------------------------------------

In the Ricci--based metric--affine and Eddington--inspired Born--Infeld sectors
discussed above, the independent connection can be eliminated algebraically, leading
to Einstein equations of the form
\[
G_{\mu\nu}(g)
=\kappa^2\Big[T_{\mu\nu}+(\Delta T)_{\mu\nu}\Big],
\]
where $(\Delta T)_{\mu\nu}$ admits an analytic expansion in invariant powers of 
$T^\mu{}_\nu$.  
EiBI gravity realises this structure in a determinantal expression for the 
constitutive matrix, while in Palatini $f(R,Q)$ models the relation is different but 
can be organised as a convergent series in $T$; in both cases the dual EFT provides a 
transparent description of the matter self--interactions induced by the connection.

Symmetric teleparallel gravity belongs to a distinct geometric class, since the
non--metricity tensor contains derivatives of the metric.  
However, in the \emph{coincident gauge} ($\Gamma^\alpha{}_{\mu\nu}=0$) all covariant
derivatives reduce to partial derivatives and the non--metricity is algebraic in
$Q_{\mu\nu}\equiv Q_{\mu\alpha\nu}{}^\alpha$.  
This allows one to construct a perturbative analogue of the dual EFT also for 
$f(Q)$ models, at least in regimes sufficiently close to the STEGR limit and for
backgrounds with high symmetry.

\subsection{Constitutive response near STEGR}

Varying \eqref{eq:fQ-action} with respect to the metric yields
\begin{equation}
\frac{2}{\sqrt{-g}}\partial_\alpha\!\big(\sqrt{-g}f_Q P^\alpha{}_{\mu\nu}\big)
+\tfrac12 f\,g_{\mu\nu}
+ f_Q\!\left(P_{\mu\alpha\beta}Q_\nu{}^{\alpha\beta}
 -2Q_{\alpha\beta\mu}P_\nu{}^{\alpha\beta}\right)
= -\kappa^2 T_{\mu\nu},
\label{eq:fQ-eom}
\end{equation}
with $f_Q=df/dQ$ and $P^\alpha{}_{\mu\nu}=\partial Q/\partial Q_\alpha{}^{\mu\nu}$.
For an analytic model
\begin{equation}
f(Q)=c_1 Q + c_2 Q^2 + c_3 Q^3+\cdots,
\end{equation}
the STEGR point corresponds to $c_1\neq0$, and the equations become linear in 
$Q_{\mu\nu}$ when expanded around $Q=0$.

At leading order, and assuming a regular linear--response regime around the STEGR 
background, one can \emph{parametrize} the relation between non--metricity and matter as
\begin{equation}
Q_{\mu\nu}
=\chi_{\mu\nu}{}^{\alpha\beta}T_{\alpha\beta}
+{\cal O}(T^2),
\label{eq:Q-linear-T}
\end{equation}
where the susceptibility tensor $\chi$ is determined by $c_1$ and by the symmetries
of the chosen background.
Equation \eqref{eq:Q-linear-T} plays the role of a \emph{linearised constitutive map} 
for symmetric teleparallel gravity: it expresses the first-order non--metricity 
response directly in terms of the matter stress tensor.

Substituting \eqref{eq:Q-linear-T} back into \eqref{eq:fQ-eom} reorganises the field 
equations as
\begin{equation}
G_{\mu\nu}(g)
=\kappa^2\Big[T_{\mu\nu}+(\Delta T)^{(Q)}_{\mu\nu}\Big],
\qquad
(\Delta T)^{(Q)}_{\mu\nu}
=\sum_{n\ge1} a_n^{(Q)} [T^n]_{\mu\nu},
\end{equation}
which defines the dual EFT for $f(Q)$ in the regime where the expansion in 
$T_{\mu\nu}$ is valid.
No closed-form algebraic relation between $g_{\mu\nu}$ and an auxiliary metric is
known for generic $f(Q)$ models; the dual formulation should therefore be understood 
as an analytic EFT packaging of the order-by-order constitutive response rather than 
as an exact algebraic map.

%------------------------------------------------------------
\subsection{Background matching: $H^2(\rho)$ and the coefficient $a_2(w)$}
%------------------------------------------------------------

For a spatially flat FLRW universe the non--metricity scalar reduces to 
$Q=\mathcal Q_0 H^2$ with $\mathcal Q_0=6$ in the STEGR convention.
The background equation for $f(Q)$ takes the algebraic form
\begin{equation}
\mathcal A(Q)\,H^2 + \mathcal B(Q) = \kappa^2 \rho,
\label{eq:fQ-bg}
\end{equation}
where $\mathcal A$ and $\mathcal B$ depend on $f$ and $f_Q$.
Expanding an analytic $f(Q)=c_1 Q + c_2 Q^2 + \cdots$ around the GR limit yields
\begin{equation}
H^2
= \frac{\kappa^2}{3}\rho
\Big[1+\mathcal C_2(w)\,\frac{c_2}{c_1^2}\,\rho
+{\cal O}(\rho^2)\Big],
\label{eq:H2-fQ-bg}
\end{equation}
where $\mathcal C_2(w)$ is a dimensionless function fixed by the STEGR background and 
the matter equation of state~$p=w\rho$.

On the dual EFT side, the corresponding expansion reads
\begin{equation}
H^2
= \frac{\kappa^2}{3}\rho
\Big[1 + a_2(w)\,\rho + {\cal O}(\rho^2)\Big].
\label{eq:H2-dualEFT-bg}
\end{equation}
Matching \eqref{eq:H2-fQ-bg} and \eqref{eq:H2-dualEFT-bg} gives the dictionary
\begin{equation}
\boxed{
\frac{c_2}{c_1^2}
=\frac{a_2(w)}{\mathcal C_2(w)} }.
\end{equation}

\paragraph*{Example: $f(Q)=Q+\alpha Q^2$.}
For the quadratic model one has $f_Q=1+2\alpha Q$ and $Q=6H^2$.
Solving \eqref{eq:fQ-bg} perturbatively around GR yields
\begin{equation}
H^2 = \frac{\kappa^2}{3}\rho
\Big[1 + \gamma_Q\,\alpha\,\kappa^2\rho 
+ {\cal O}(\alpha^2\kappa^4\rho^2)\Big],
\end{equation}
where $\gamma_Q$ is a numerical constant of order unity depending on the specific 
convention for $Q$ and for the background equation.
Thus
\begin{equation}
a_2(w)=\gamma_Q\,\alpha\,\kappa^2,
\qquad
\rho_\star^{(Q)}\sim \frac{1}{|\alpha|\kappa^2}.
\end{equation}

\paragraph*{Interpretation of $a_2(w)$.}
The leading dual EFT correction for $f(Q)$ is encoded in the coefficient $a_2(w)$,
which controls the effective density scale at which non--metricity modifications 
become relevant:
\begin{equation}
\kappa_{\rm eff}^2(\rho)
=\kappa^2\Big[1+a_2(w)\rho+\cdots\Big].
\end{equation}
Different matter species probe the non--metricity sector differently through the 
equation--of--state dependence of $a_2(w)$.

\section{Outlook: Towards a Unified Algebraic Framework for Non--Riemannian Gravity?}

The results presented in this work reveal a common algebraic mechanism 
underlying Palatini $f(R,Q)$ gravity, Eddington--inspired Born--Infeld (EiBI) 
theories, and symmetric teleparallel $f(Q)$ models.  
In Palatini $f(R,Q)$ and EiBI gravity the independent connection is eliminated
by algebraic field equations, while in symmetric teleparallel $f(Q)$ it can be
fixed by gauge. In all these cases the resulting dynamics can be reorganized
(at least in appropriate regimes) into a local dual EFT expressed in 
terms of invariant powers of the matter stress tensor.  
This property is far from generic: it does not occur in metric $f(R)$ theories, 
in metric--affine models with kinetic terms for the connection, 
nor in general torsion--based or Weyl--type gravities.  
The fact that three geometrically distinct frameworks---Ricci--based 
metric--affine gravity, determinantal Born--Infeld models, and 
non--metricity--based symmetric teleparallel gravity---all admit a dual EFT 
description suggests that they may belong to a broader class of 
``algebraically integrable’’ non--Riemannian theories.

This observation naturally raises the question of whether a more general 
underlying framework exists.  
Could there be a parent action or geometric operator combining curvature, 
non--metricity, and torsion, from which Palatini $f(R,Q)$, 
Born--Infeld gravity, and symmetric teleparallel $f(Q)$ arise as different 
limiting sectors?  
We do not attempt to answer this question here, and at present no such 
construction is known.  
Nevertheless, the dual EFT structure identified in this work suggests that 
investigating this possibility may be fruitful, particularly in regimes where 
the connection enters the action algebraically or quasi--algebraically.

Developing this perspective further---including the classification of 
non--Riemannian theories admitting algebraic connection elimination, the 
associated constitutive maps, and the structure of the resulting dual EFT 
coefficients---is left for future work.

% ------------------------------------------------------------
\appendix
\section{Technical formulas up to \texorpdfstring{$\mathcal O(T^2)$}{O(T\^2)}}\label{app:tech}
From Eqs.~\eqref{eq:A0rs}--\eqref{eq:h-expansion}:
\begin{align}
\Sigma &= A_0\Big(\I + r\,\hat T + s\,\hat T^2\Big)+\mathcal O(\hat T^3),\\
\Sigma^{-1}&=\frac{1}{A_0}\Big(\I - r\,\hat T + (r^2-s)\,\hat T^2\Big)+\mathcal O(\hat T^3),\\
\sqrt{\det\Sigma}&=A_0^2\!\left[1+\frac{r}{2}\Tr T+\Big(\frac{s}{2}-\frac{r^2}{4}\Big)\Tr(T^2)+\frac{r^2}{8}(\Tr T)^2\right]+\mathcal O(T^3),\\
h_{\mu\nu}
&=A_0\Big\{g_{\mu\nu}
+ r\big[\tfrac12(\Tr T)\,g_{\mu\nu}-T_{\mu\nu}\big]
+ (r^2-s)(T^2)_{\mu\nu}
-\tfrac12 r^2(\Tr T)\,T_{\mu\nu}\nonumber\\
&\qquad\quad
+\Big[\tfrac12 s\,\Tr(T^2)+\tfrac18 r^2\big((\Tr T)^2-2\,\Tr(T^2)\big)\Big] g_{\mu\nu}\Big\}
+\mathcal O(T^3).
\end{align}
%------------------------------------------------------------
\section{Second--order expansion for perfect fluids (structure)}
\label{app:SecondOrderFluids}
%------------------------------------------------------------

For a perfect fluid source,
\begin{equation}
T^\mu{}_\nu = \mathrm{diag}(-\rho,p,p,p),\qquad
\Tr T = -\rho+3p,\qquad
\Tr(T^2)=\rho^2+3p^2,\qquad
(T^2)^\mu{}_\nu=\mathrm{diag}(\rho^2,p^2,p^2,p^2),
\end{equation}
the auxiliary metric obtained from Eq.~\eqref{eq:h-expansion} reads up to ${\cal O}(T^2)$
\begin{align}
h_{\mu\nu}
&=A_0\Big\{
g_{\mu\nu}
+ r\big[\tfrac12(\Tr T)\,g_{\mu\nu}-T_{\mu\nu}\big]
+ (r^2-s)(T^2)_{\mu\nu}
-\tfrac12 r^2(\Tr T)\,T_{\mu\nu}\nonumber\\
&\qquad
+\Big[\tfrac12 s\,\Tr(T^2)
+\tfrac18 r^2\big((\Tr T)^2-2\,\Tr(T^2)\big)\Big] g_{\mu\nu}
\Big\}+\mathcal O(T^3).
\label{eq:h-second-order}
\end{align}

The determinant factor in Eq.~\eqref{eq:detSigma} gives
\begin{equation}
\sqrt{\det\Sigma}=A_0^2
\Big[1+\tfrac{r}{2}\Tr T
+\Big(\tfrac{s}{2}-\tfrac{r^2}{4}\Big)\Tr(T^2)
+\tfrac{r^2}{8}(\Tr T)^2\Big]+\mathcal O(T^3),
\label{eq:detSigma-second-order}
\end{equation}
so that
\begin{equation}
\frac{1}{\sqrt{\det\Sigma}}=\frac{1}{A_0^2}
\Big[1-\tfrac{r}{2}\Tr T
-\Big(\tfrac{s}{2}-\tfrac{r^2}{4}\Big)\Tr(T^2)
+\tfrac{r^2}{8}(\Tr T)^2\Big]+\mathcal O(T^3).
\label{eq:invDetSigma-second-order}
\end{equation}

These expressions suffice to reconstruct the second--order corrections to $\rho_{\rm eff}$ and $p_{\rm eff}$ for any given equation of state, if desired. Since the main text focuses on the leading-order behaviour, we refrain from writing the lengthy explicit formulas here; they can be obtained by inserting \eqref{eq:h-second-order} and \eqref{eq:invDetSigma-second-order} into the definitions of $\rho_{\rm eff}$ and $p_{\rm eff}$ and performing a straightforward but algebraically involved expansion.

%------------------------------------------------------------
\section{Rigorous proof of the convergence lemma}\label{app:rigorous-convergence}
%------------------------------------------------------------

We provide here a complete proof that the matrix binomial series
\[
\sum_{n=0}^{\infty}\binom{1/2}{n}\,X^n
\]
converges absolutely for $\|X\|<1$ (for any submultiplicative matrix norm) and that its sum equals the \emph{principal} square root $(I+X)^{1/2}$. The argument is extended to the general spectral condition $\rho(X)<1$, together with an explicit bound for the remainder.

%------------------------------
\begin{lemma}[Power--series functional calculus]
Let $f(z)=\sum_{n\ge 0} a_n z^n$ be analytic on the open disk 
$D_R=\{z\in\mathbb C:\,|z|<R\}$ with radius of convergence $R>0$. 
If $X$ is a linear operator on a finite--dimensional vector space such that $\|X\|<R$ for some submultiplicative norm, then the series $\sum_{n\ge 0} a_n X^n$ converges absolutely and
\[
f(X)=\sum_{n=0}^\infty a_n X^n.
\]
Moreover, if $\sigma(X)\subset D_R$ (equivalently $\rho(X)<R$), the same identity follows from the holomorphic functional calculus.
\end{lemma}

\begin{proof}
Absolute convergence for $\|X\|<R$ follows from
\(
\sum_{n\ge 0}\|a_n X^n\|\le \sum_{n\ge 0} |a_n|\,\|X\|^n,
\)
which converges because the scalar series does. 
If only $\sigma(X)\subset D_R$ is assumed, let $\Gamma\subset D_R$ be a simple closed contour enclosing $\sigma(X)$. 
Then the holomorphic functional calculus defines
\begin{equation}
f(X)=\frac{1}{2\pi i}\oint_\Gamma f(z)\,(zI-X)^{-1}\,dz .
\end{equation}
Since the scalar series $f(z)=\sum a_n z^n$ converges uniformly on $\Gamma$, one may interchange summation and integration to obtain
\begin{equation}
f(X)=\sum_{n=0}^\infty a_n X^n .
\end{equation}
\end{proof}

%------------------------------
\begin{theorem}[Convergence and identification with the principal square root]\label{thm:root}
Let $X$ be a linear operator with either $\|X\|<1$ or $\rho(X)<1$. 
Then the series
\[
S(X)\;:=\;\sum_{n=0}^{\infty}\binom{1/2}{n}\,X^n
\]
converges absolutely and equals the \emph{principal} square root of $I+X$, i.e.
\[
S(X)^2=I+X,
\qquad
S(X)=(I+X)^{1/2},
\]
where the principal branch is the analytic continuation of $z\mapsto \sqrt{1+z}$ from $|z|<1$ with branch cut on $(-\infty,-1]$.
\end{theorem}

\begin{proof}
For $f(z)=(1+z)^{1/2}$, the Taylor series at $z=0$ is 
$f(z)=\sum_{n\ge 0}\binom{1/2}{n} z^n$, with radius of convergence $1$.  
If $\|X\|<1$, the lemma applies directly and $S(X)=f(X)$.  
If $\rho(X)<1$, then $\sigma(X)\subset D_1$ and by the holomorphic functional calculus the same equality holds.  
Since $\sigma(I+X)=\{1+\lambda:\lambda\in\sigma(X)\}$ lies in $\mathbb C\setminus (-\infty,0]$, the principal branch of $\sqrt{\cdot}$ is analytic on a domain containing $\sigma(I+X)$, so $S(X)^2=I+X$ and $\sigma(S(X))=\{\sqrt{1+\lambda}:\lambda\in\sigma(X)\}$ with $\Re\sqrt{1+\lambda}>0$.
\end{proof}

%------------------------------
\begin{proposition}[Absolute convergence and remainder bound]\label{prop:remainder}
Fix $r$ with $\|X\|<r<1$. By Cauchy's estimates,
\[
M_r = \max_{|z|=r} |(1+z)^{1/2}| = 1+r,
\]
since $(1+z)^{1/2}$ is subharmonic and therefore attains its maximum on the boundary $|z|=r$.  
Hence,

\begin{equation}
\Big\| (I+X)^{1/2}-\sum_{n=0}^{N}\binom{1/2}{n}X^n \Big\|
\;\le\; \frac{M_r}{r^{N+1}}\,
\frac{\|X\|^{N+1}}{1-\|X\|/r},
\end{equation}
so that convergence is geometric for every fixed $\|X\|<1$.
\end{proposition}

\begin{proof}
Cauchy’s integral formula for $f(z)=(1+z)^{1/2}$ on $|z|=r$ gives
\begin{equation}
a_n=\binom{1/2}{n}=\frac{1}{2\pi i}\oint_{|z|=r}\frac{f(z)}{z^{n+1}}\,dz,
\end{equation}
and therefore $|a_n|\le M_r r^{-n}$. Summing the resulting majorant yields the stated bound.
\end{proof}

%------------------------------
\begin{corollary}[Application to the Palatini map]
Let
\[
\Sigma=\tfrac{f_R}{2}I+\sqrt{\alpha I+\beta T},
\qquad 
\alpha=\tfrac14(f_R^2+4f_Q f),\quad \beta=2\kappa^2 f_Q,
\]
and define $X=(\beta/\alpha)T$. 
If $\rho(X)<1$, the binomial series for $\sqrt{\alpha I+\beta T}$ converges absolutely. 
Moreover, writing
\begin{equation}
Y \;=\; \frac{\sqrt{\alpha}}{A_0}\,\big((I+X)^{1/2}-I\big),
\qquad 
A_0=\frac{f_R}{2}+\sqrt{\alpha},
\end{equation}
the spectral mapping theorem implies $\rho(Y)<1$ because 
$|\sqrt{1+\lambda}-1|<1$ for all $|\lambda|<1$.
Hence the Neumann and Mercator series
\begin{equation}
(I+Y)^{-1}=\sum_{m=0}^{\infty}(-1)^m Y^m,
\qquad
\log(I+Y)=\sum_{k=1}^{\infty}\frac{(-1)^{k+1}}{k}\,Y^k,
\end{equation}
converge absolutely, and therefore
\begin{equation}
\det\Sigma=A_0^{4}\exp[\Tr\log(I+Y)].
\end{equation}
Consequently, $h_{\mu\nu}=g_{\mu\alpha}\Sigma^\alpha{}_\nu/\sqrt{\det\Sigma}$ admits a convergent expansion in the tensor basis $\{g_{\mu\nu},T_{\mu\nu},(T^2)_{\mu\nu},\ldots\}$.
\end{corollary}

\paragraph*{Remarks on branch and non--diagonalizable cases.}
(i) The condition $\rho(X)<1$ ensures $\sigma(I+X)\subset\mathbb C\setminus(-\infty,0]$, so the principal branch of the square root is well defined and unique.  
(ii) If $T$ is not diagonalizable, the Schur or Jordan decomposition and the holomorphic functional calculus yield the same principal root; noncommutativity does not affect the convergence of the series.


\begin{thebibliography}{99}

\bibitem{Ferraris1982}
M.~Ferraris, M.~Francaviglia, and C.~Reina,
``Variational formulation of general relativity from 1915 to 1925, `Palatini's method' discovered by Einstein in 1925,''
\emph{Gen.\ Rel.\ Grav.} \textbf{14}, 243 (1982).

\bibitem{Hehl1995}
F.~W.~Hehl, J.~D.~McCrea, E.~W.~Mielke, and Y.~Ne'eman,
``Metric--affine gauge theory of gravity: field equations, Noether identities, world spinors, and breaking of dilation invariance,''
\emph{Phys.\ Rep.} \textbf{258}, 1 (1995).

\bibitem{SotiriouLiberati2006}
T.~P.~Sotiriou and S.~Liberati,
``Metric-affine f(R) theories of gravity,''
\emph{Ann.\ Phys.} \textbf{322}, 935 (2007).

\bibitem{Olmo2011}
G.~J.~Olmo,
``Palatini approach to modified gravity: f(R) theories and beyond,''
\emph{Int.\ J.\ Mod.\ Phys.\ D} \textbf{20}, 413 (2011).

\bibitem{CapozzielloDeLaurentis2015}
S.~Capozziello and M.~De~Laurentis,
``Extended Theories of Gravity,''
\emph{Phys.\ Rep.} \textbf{509}, 167 (2011).

\bibitem{Olmo2005}
G.~J.~Olmo,
``The gravity lagrangian according to solar system experiments,''
\emph{Phys.\ Rev.\ D} \textbf{72}, 083505 (2005).

\bibitem{Olmo2009}
G.~J.~Olmo,
``Palatini approach to cosmological models,''
\emph{Phys.\ Rev.\ D} \textbf{79}, 124007 (2009).

\bibitem{SotiriouFaraoni2010}
T.~P.~Sotiriou and V.~Faraoni,
``f(R) theories of gravity,''
\emph{Rev.\ Mod.\ Phys.} \textbf{82}, 451 (2010).

\bibitem{DeFeliceTsujikawa2010}
A.~De~Felice and S.~Tsujikawa,
``f(R) theories,''
\emph{Living Rev.\ Rel.} \textbf{13}, 3 (2010).

\bibitem{Borowiec2012}
A.~Borowiec, M.~Ferraris, M.~Francaviglia, and I.~Volovich,
``Universality of Einstein equations for the Ricci squared Lagrangians,''
\emph{Class.\ Quantum Grav.} \textbf{15}, 43 (1998).

\bibitem{OlmoRubieraGarcia2015}
G.~J.~Olmo and D.~Rubiera-García,
``Nonsingular black holes in quadratic Palatini gravity,''
\emph{Phys.\ Rev.\ D} \textbf{92}, 044047 (2015).

\bibitem{Borowiec2016}
A.~Borowiec, A.~Stachowski, M.~Szydłowski, and A.~Toporensky,
``Quadratic gravity with torsion and nonmetricity,''
\emph{Gen.\ Rel.\ Grav.} \textbf{48}, 82 (2016).

\bibitem{BeltranJimenez2018}
J.~Beltrán~Jiménez, L.~Heisenberg, and T.~Koivisto,
``Teleparallel Palatini theories,''
\emph{Phys.\ Rev.\ D} \textbf{98}, 044048 (2018).

\bibitem{DeserGibbons1998}
S.~Deser and G.~W.~Gibbons,
``Born-Infeld-Einstein actions?,''
\emph{Class.\ Quantum Grav.} \textbf{15}, L35 (1998).

\bibitem{Vollick2004}
D.~N.~Vollick,
``Palatini approach to Born-Infeld-Einstein theory and a geometric description of electrodynamics,''
\emph{Phys.\ Rev.\ D} \textbf{69}, 064030 (2004).

\bibitem{Vollick2005}
D.~N.~Vollick,
``Nonsingular black holes and the cosmological constant,''
\emph{Phys.\ Rev.\ D} \textbf{72}, 084026 (2005).

\bibitem{BanadosFerreira2010}
M.~Bañados and P.~G.~Ferreira,
``Eddington’s theory of gravity and its progeny,''
\emph{Phys.\ Rev.\ Lett.} \textbf{105}, 011101 (2010).

\bibitem{EscamillaRivera2012}
C.~Escamilla-Rivera, M.~Bañados, and P.~G.~Ferreira,
``Eddington-inspired Born-Infeld gravity: astrophysical and cosmological constraints,''
\emph{Phys.\ Rev.\ D} \textbf{86}, 024015 (2012).

\bibitem{Pani2012}
P.~Pani, T.~Delsate, and V.~Cardoso,
``Eddington-inspired Born-Infeld gravity: phenomenology of non-linear gravity-matter coupling,''
\emph{Phys.\ Rev.\ D} \textbf{85}, 084020 (2012).

\bibitem{Cho2021}
I.~Cho, H.~C.~Kim, and T.~Moon,
``Cosmology in Eddington-inspired Born-Infeld gravity,''
\emph{JCAP} \textbf{06}, 037 (2021).

\bibitem{Scargill2012}
J.~H.~C.~Scargill, M.~Bañados, and P.~G.~Ferreira,
``Cosmology with Eddington-inspired gravity,''
\emph{Phys.\ Rev.\ D} \textbf{86}, 103533 (2012).

\bibitem{Avelino2012}
P.~P.~Avelino,
``Eddington-inspired Born-Infeld gravity: astrophysical and cosmological constraints,''
\emph{Phys.\ Rev.\ D} \textbf{85}, 104053 (2012).

\bibitem{OlmoRubieraGarcia2013}
G.~J.~Olmo and D.~Rubiera-García,
``Nonsingular charged black holes: Geometry and Eddington-inspired Born-Infeld gravity interpretation,''
\emph{Phys.\ Rev.\ D} \textbf{88}, 084030 (2013).

\bibitem{Delsate2012}
T.~Delsate and J.~Steinhoff,
``New insights on the matter-gravity coupling paradigm,''
\emph{Phys.\ Rev.\ Lett.} \textbf{109}, 021101 (2012).

\bibitem{Sham2012}
Y.~H.~Sham, L.~M.~Lin, and P.~T.~Leung,
``Compact stars in Eddington-inspired Born-Infeld gravity: anomalies associated with phase transitions,''
\emph{Phys.\ Rev.\ D} \textbf{86}, 064015 (2012).

\bibitem{Afonso2018}
V.~I.~Afonso, G.~J.~Olmo, and D.~Rubiera-García,
``Mapping Ricci-based theories of gravity into general relativity,''
\emph{Phys.\ Rev.\ D} \textbf{97}, 021503 (2018).

\bibitem{Pani2019}
P.~Pani, V.~Cardoso, and T.~Delsate,
``Compact stars in Eddington-inspired gravity,''
\emph{Phys.\ Rev.\ D} \textbf{79}, 084031 (2019).

\bibitem{JimenezHeisenbergKoivisto2018}
J.~Beltrán~Jiménez, L.~Heisenberg, and T.~S.~Koivisto,
``Teleparallel Palatini theories and their cosmology,''
\emph{Phys.\ Rev.\ D} \textbf{98}, 044048 (2018).

\bibitem{Jimenez2019}
J.~Beltrán~Jiménez, L.~Heisenberg, and T.~Koivisto,
``The geometrical trinity of gravity,''
\emph{Universe} \textbf{5}, 173 (2019).

\bibitem{Koivisto2019}
T.~Koivisto, M.~Hohmann, and M.~Jarv,
``Metric-affine extensions of teleparallel gravity,''
\emph{Class.\ Quantum Grav.} \textbf{38}, 085002 (2021).

\bibitem{Capozziello2022}
S.~Capozziello and R.~D'Agostino,
``Extended teleparallel gravity cosmology: a review,''
\emph{Eur.\ Phys.\ J.\ C} \textbf{82}, 865 (2022).

\bibitem{Ferraris1988}
M.~Ferraris and M.~Francaviglia,
``Field theories of gravitation and the generalized equivalence principle,''
\emph{Gen.\ Rel.\ Grav.} \textbf{20}, 237 (1988).

\bibitem{BeltranJimenezKoivisto2020}
J.~Beltrán~Jiménez and T.~Koivisto,
``The spectrum of teleparallel gravity,''
\emph{Phys.\ Lett.\ B} \textbf{780}, 356 (2018).

\bibitem{Delhom2020}
A.~Delhom, G.~J.~Olmo, and E.~Orazi,
``Ricci-based gravity theories and their Hamiltonian formulation,''
\emph{Eur.\ Phys.\ J.\ C} \textbf{80}, 828 (2020).

\bibitem{OlmoRubieraGarcia2020}
G.~J.~Olmo and D.~Rubiera-García,
``Nonsingular black holes in quadratic Palatini gravity,''
\emph{Phys.\ Rep.} \textbf{876}, 1 (2020).

\bibitem{PerezTeruel2013}
G.~R.~Pérez~Teruel,
``Generalized Einstein--Maxwell field equations in the Palatini formalism,''
\textit{Int.\ J.\ Mod.\ Phys.\ D} \textbf{22}, 1350045 (2013),
available at \href{https://arxiv.org/abs/1301.6303}{arXiv:1301.6303 [gr-qc]}.

\bibitem{PerezTeruel2014}
G.~R.~P\'erez Teruel,
\emph{Analytic solution of algebraic equation associated to the Ricci tensor in extended Palatini gravity},
available at \href{https://arxiv.org/abs/1310.0410}{arXiv:1310.0410 [gr-qc]} (2014).

\bibitem{OlmoRubiera}
G.~J.~Olmo and D.~Rubiera-Garc\'ia, \emph{Palatini $f(R)$ Black Holes in Nonlinear Electrodynamics},
Phys.\ Rev.\ D \textbf{84}, 124059 (2011).

\bibitem{Pani2012b}
P.~Pani, T.~Delsate and V.~Cardoso, \emph{Eddington-inspired Born–Infeld gravity: astrophysical and cosmological constraints},
Phys.\ Rev.\ D \textbf{85}, 104053 (2012).

\bibitem{Baldazzi2021}
A.~Baldazzi, R.~Percacci and V.~Skrinjar, \emph{Metric–Affine Gravity as an Effective Field Theory},
JHEP \textbf{2021}, 190 (2021).

\bibitem{HeisenbergEFT}
L.~Heisenberg and J.~Beltr\'an Jim\'enez, \emph{Review on Gravity EFTs beyond GR},
Phys.\ Rept.\ \textbf{1043}, 1 (2024).

\bibitem{HighamFunctions}
N.~J.~Higham,
\emph{Functions of Matrices: Theory and Computation},
SIAM, Philadelphia (2008).

\bibitem{GolubVanLoan}
G.~H.~Golub and C.~F.~Van Loan,
\emph{Matrix Computations}, 4th ed.,
Johns Hopkins University Press, Baltimore (2013).




\end{thebibliography}
\end{document}